\documentclass{article}

\usepackage{hyperref}
\usepackage{url}
\usepackage{amssymb}
\usepackage{amsmath}
\usepackage{amsthm}
\usepackage{color}
\usepackage{xcolor}
\usepackage{fullpage}
\usepackage{xspace}
\usepackage[linesnumbered,ruled,vlined]{algorithm2e}
\usepackage{color}
\usepackage{multirow}
\usepackage{nccmath}
\usepackage[normalem]{ulem}
\usepackage{cite}
\usepackage{psfig}
\usepackage{epsfig}
\usepackage{psfrag}       

\newcommand{\LE}{Leader Election\xspace}
\newcommand{\eLE}{Revocable Leader Election\xspace}
\newcommand{\sLE}{Irrevocable Leader Election\xspace}
\newcommand{\eLEname}{Blind \LE with Certificates via Diffusion with Thresholds\xspace}
\newcommand{\congest}{{\sc Congest}\xspace}

\newcommand{\CMC}{{\sc Avg}\xspace} 

\newcommand{\ANs}{Anonymous Networks\xspace}

\newcommand{\dk}[1]{{\color{red} #1}}

\newcommand{\mm}[1]{\textcolor{blue}{#1}}

\renewcommand{\mm}[1]{#1}
\renewcommand{\dk}[1]{#1}
\newcommand{\remove}[1]{}

\newcommand{\ldr}{{leader}\xspace}

\newcommand{\cand}{{white}\xspace}

\newcommand{\notcand}{{black}\xspace}
\newcommand{\dist}{{diffusion}\xspace}
\newcommand{\Dist}{{Diffusion}\xspace}
\newcommand{\diss}{{dissemination}\xspace}
\newcommand{\Diss}{{Dissemination}\xspace}
\newcommand{\try}{{certification}\xspace}
\newcommand{\Try}{{Certification}\xspace}
\newcommand{\dec}{{decision}\xspace}
\newcommand{\Dec}{{Decision}\xspace}
\newcommand{\iso}{{isoperimetric number}\xspace}

\newcommand{\rec}{{rec}}
\newcommand{\trans}{{trans}}

\newcommand{\cE}{{{\cal E}}\xspace}
\newcommand{\cA}{{{\cal A}}\xspace}

\newtheorem{theorem}{Theorem}
\newtheorem{lemma}{Lemma}
\newtheorem{corollary}{Corollary}

\newtheorem{definition}{Definition}

\begin{document}

\title{
Time and Communication Complexity of\\ \LE in \ANs
}

\author{
Dariusz R. Kowalski\\
Augusta University, Augusta, GA, USA,\\ 
SWPS University, Warsaw, Poland.\\
\href{mailto:dkowalski@augusta.edu}{dkowalski@augusta.edu}
\and
Miguel A. Mosteiro\\
Pace University,\\ New York, NY, USA.\\
\href{mailto:mmosteiro@pace.edu}{mmosteiro@pace.edu}
}

\date{}

\maketitle

\begin{abstract}
We study the problem of randomized \LE in synchronous distributed networks with indistinguishable nodes. We consider algorithms that work on networks of arbitrary topology in two 
settings,
depending on whether the size of the network, i.e., the number of nodes $n$, is known or not. 
In the former setting, we present a new \LE protocol that improves over previous work
by lowering message complexity and making it 
\mm{close to a lower bound by a factor in~$\widetilde{O}(\sqrt{t_{mix}\sqrt{\Phi}})$, where $\Phi$ is the conductance and $t_{mix}$ is the mixing time of the network graph.}
We then show that lacking the network size no \LE algorithm can guarantee that the election is final with constant probability, even with unbounded communication. Hence, we further classify the problem as \sLE (the classic one, requiring knowledge of $n$ -- as is our first protocol) or \eLE, and present a new polynomial time and message complexity \eLE algorithm in the setting without knowledge of network size.
We analyze time and message complexity of our protocols in the \congest model of communication.
%
\end{abstract}


\section{Introduction}
\label{sec:intro}

\LE is one of the most studied problems in distributed computing. The aim is to select exactly one of the processors participating in the computing platform. For instance, in computer networks, to select a unique node to adopt a special role. The \LE problem is important for distributed computing because having a distinguished processor is crucial for many applications, such as resource allocation, scheduling, synchronization and any other problems where symmetry among processors needs to be broken to make decisions. 
Studies of \LE include deterministic and randomized algorithms, explicit (all nodes know who is the leader) and implicit versions, and algorithms with bounded running time or eventual termination~\cite{lynchbook,attiyabook}. 

Motivated by recent developments for massive Ad-hoc Networks, possibly embedded in the Internet of Things, in this work we study \LE in \ANs. That is, networks where nodes do not have identifiers (ID), labels, or any other form of distinguishing among themselves initially. Given that it is not possible to solve \LE deterministically in such scenario~\cite{angluin1980local} we focus on randomized algorithms with provable performance guarantees. Given that energy is usually a crucial limited resource in Ad-hoc Networks, we evaluate our protocols by their time and message complexity, limiting the per-round communication through each link to $O(\log n)$ bits, where $n$ is the number of network nodes, known as the \congest~model. 

A usual approach in an environment where nodes are indistinguishable is to choose IDs at random. To guarantee uniqueness the probability is a function of $n$. However, 
\mm{in Ad-hoc Networks composed by inexpensive nodes deployed in massive numbers, it may be very difficult or inconvenient to determine their number.}
Hence, we study \LE algorithms for two models of \ANs: with known or unknown network size. As we show, not knowing $n$ there is no algorithm that elects a leader and stops, even with constant probability. Intuitively, the reason is that not knowing $n$ it is always possible to hear in the future from other nodes located far away. Nevertheless, it is possible to elect a local leader based on the knowledge gained so far and change such choice in the future. Motivated by this intuition, we further specify 
\mm{the problem as \eLE, naming the standard problem as \sLE when necessary for clarity.}
In both versions the time to elect the final \ldr is bounded, 
\mm{but in \eLE the algorithm is not required to stop, since nodes do not know if further knowledge of the network gained later may change decisions already taken.} 


\paragraph{Contributions} 


For \ANs where 
the network size $n$ is known, we present a \LE algorithm that elects a unique leader using $\widetilde{O}(\sqrt{n t_{mix}/\Phi})$~\footnote{Throughout, we use the asymptotic notation $\widetilde{O}(\cdot)$ to mean that $f(n) \in \widetilde{O}(g(n))$ if $\exists k : f(n) \in O(g(n) \log^k g(n))$, as well as the analogous $\widetilde{\Theta}(\cdot)$ and $\widetilde{\Omega}(\cdot)$.} point to point messages/bits of communication in the \congest model, with high probability,\footnote{Throughout, we say that a stochastic event holds \emph{with high probability}, or \emph{whp} for short, if it holds with probability at least $1-1/n^c$, for some constant $c>0$.} 
where $t_{mix},\Phi$ are the network mixing time and \mm{graph conductance (as defined in Section~\ref{sec:model})}, respectively.
The algorithm works in time $O(t_{mix}\log^2 n)$, knowing conductance and mixing time (cf. Theorem~\ref{thm:knownn}). 
\mm{
In a PODC 2018 paper, Gilbert et al.~\cite{GilbertRSpodc2018} showed that
$O( t_{mix}\sqrt{n} \log^{7/2} n)$ messages are enough to solve implicit (Irrevocable) \LE with high probability in $O(t_{mix} n^2)$ time, for known network size $n$. 
With respect to that work, our \sLE algorithm improves on messages while having asymptotically the same time complexity.
}
\dk{The former holds because} the term $\sqrt{n t_{mix}/\Phi}$ is asymptotically smaller than $\sqrt{n} t_{mix}$,
for $t_{mix} = \omega(1/\Phi)$, and never asymptotically larger, due to 
the known bounds $1/\Phi \leq t_{mix} \leq 1/\Phi^2$~\cite{montenegro2006mathematical}. 
Our algorithm is also nearly optimal for networks with a reasonable expansion, $t_{mix} = \widetilde{\Theta}(1/\Phi)\le \widetilde{O}(D)$, as~for~such networks it achieves time $\widetilde{O}(D)$ and message complexity
$\widetilde{O}(\sqrt{n t_{mix}/\Phi})\le \widetilde{O}(\sqrt{n}/\Phi)$, where the latter is away from the lower bound $\Omega(\sqrt{n}/\Phi^{3/4})$
in~\cite{GilbertRSpodc2018} by a factor~$\widetilde{O}(1/\Phi^{1/4})$.

\dk{In order to accomplish the task with a nearly optimal communication, our algorithm relies on a novel technique of cautious broadcast.
In short, a spanning tree is being created by randomly selected candidates \mm{for leadership} with a random ID, in a distributed way but under restriction that
only nodes in less populated branches are given permit to extend the tree (by selecting a random neighbor).
This technique allows to create ``territories'' which could be efficiently discovered by a small number of independent random walks
initiated by other candidates; reversing the cautious broadcast allows to learn random IDs of other candidates, and 
the one that does not receive any candidate with a bigger (random) ID becomes leader, all with high probability.}

For \ANs with unknown $n$, we show that, 
for any non-decreasing positive integer function $T(n)$ and a constant $0<c<1$,
there is no algorithm that solves the (Irrevocable) \LE problem in time $T(n)$ with probability $c$ (cf. Theorem~\ref{thm:imposs}).
\mm{To prove such claim we use a probabilistic pumping-wheel technique. In brief, we show that in a given cycle formed by enough number of instances of a path of nodes called a \emph{witness,} with large enough probability, the information \emph{pumped} towards the center of some witness is ``not enough''. Consequently, nodes in both sides of the center wrongly elect leaders separately.} 

In face of such impossibility, we designed an algorithm that solves the \eLE problem in the \congest model with high probability with time and message complexities in \mm{$\widetilde{O}(n^{4(2+\epsilon)})$ and $\widetilde{O}(n^{4(2+\epsilon)} m )$} respectively, for any $\epsilon>0$, where $m$ is the number of network links 
(cf. Corollary~\ref{cor:unknownN}).  
%
\mm{
Our algorithm, that solves the explicit version of \eLE, does not use any knowledge of the network. For settings where the isoperimetric number of the network graph $i(G)$ (as defined in Section~\ref{sec:model}) is known to the algorithm, we show tighter bounds of 
$\widetilde{O}(n^{4(1+\epsilon)}/i(G)^2)$ 
time and 
$\widetilde{O}(n^{4(1+\epsilon)}m/i(G)^2)$ 
messages
(cf. Theorem~\ref{thm:unknownN}). 
}

\mm{
We call this algorithm \eLEname because it embeds two novel techniques. On one hand, to decide leadership \emph{blindly} (i.e. without knowledge of the network) we combine each node ID chosen at random with an estimate of the network size used to define the sample space of such choice. The latter works as a \emph{certificate} of the validity of such ID. On the other hand, nodes probe each estimate through a \emph{diffusion} process of some node potential values. The process includes various alarms and \emph{thresholds} to detect whether the estimate is close enough to the network size or not. 
}

Our study of \eLE protocols is the first one to the best of our knowledge.

\paragraph{Roadmap} The rest of the paper is organized as follows. In the next section we formally define the problems studied, the model, and necessary notation and definitions. 
We overview previous work in Section~\ref{sec:relwork}. 
We present our \sLE protocol using network size in Section~\ref{known-n}. 
We prove the impossibility of \sLE without using network size and present our \eLEname protocol and its analysis in Section~\ref{sec:unknownN}. 
We finish with conclusions and open problems. 

\section{Model, Problem, and Notation}
\label{sec:model}

We consider a communication network where every pair of nodes are able to communicate, possibly through multiple hops. We assume that communication links and network nodes do not fail after startup. We model the topology of such network with a connected undirected graph $G=(V,E)$, where $V$ is the set of nodes and $E$ is the set of links, each of them given as an unordered pair of nodes. We say that nodes $u$ and $v$ are \emph{neighbors} if $(u,v)\in E$.
We denote $n=|V|$ and $m=|E|$.

We assume that time is slotted in \emph{rounds} of communication, and we assume that the network is \emph{globally synchronous}, that is, the time length of a round is the same for all nodes, and that all nodes can start the execution of any given protocol simultaneously. We model communication limits using the \congest model~\cite{peleg2000distributed}. That is, in each synchronous round, each node can send a \emph{message} of size $O(\log n)$ bits to each of its neighbors. We assume that local computations take negligible time --
we measure time performance as the number of rounds of communication. 

We assume that nodes do not have default identities (ID), but they do have some labeling of the communication links to neighboring nodes, 
known as~\emph{port~numbers}.

\mm{
With respect to the randomized \LE problem, we adopt the following definition from~\cite{GilbertRSpodc2018}.
\begin{definition}
Every node of a given distributed network has a flag boolean variable initialized to false and, after the process of election, with high probability, only one node, called the leader, raises its flag by setting the flag variable to true. An algorithm is said to solve \LE in $t$ rounds, if within $t$ rounds nodes elect a unique leader with high probability, and all nodes stop the election process after $t$ rounds.
\end{definition}
\LE is usually further specified as \emph{explicit} when all nodes learn which node is the leader, or \emph{implicit} otherwise.
In this work, we further specify \LE to allow nodes to continue executing the algorithm, as follows.
\begin{definition}
Every node of a given distributed network has a flag boolean variable initialized to false and, after the process of election, with high probability, only one node, called the leader, raises its flag by setting the flag variable to true. An algorithm is said to solve \eLE in $t$ rounds, if within $t$ rounds nodes elect a unique leader with high probability.
\end{definition}
The motivation behind this new definition (and the name ``revocable'') is to allow nodes to change a decision about which node is the leader. For instance, if more than one leader is elected and nodes detect the error. The rationale is that nodes may never know whether their leadership decision is final or will be changed in the future (for instance if the number of nodes is unknown). Nevertheless, the time to elect the final leader must be bounded, even if nodes do not know such bound.   
}

\mm{
To avoid confusion with the new definition, throughout the paper we denote the classic version of \LE where nodes must stop as \emph{\sLE}.
}

Throughout the paper, we use the following standard definitions. 
%

We define the \emph{mixing time} of an $n$-node graph $G$, $t_{mix}(G)$ 
as the minimum $t$ such that, for each $\pi_0$, $||P\pi_t-\pi_{*}||_\infty \leq 1/(2n)$, where 
$P$ is the transition matrix of a random walk in graph $G$,
$\pi$ is the probability distribution of the random walk over the nodes of $G$,
$\pi_0$ is the starting distribution,
$\pi_{*}$ is the stationary distribution,
and $||\cdot||_\infty$ denotes the maximum norm on a vector.

\mm{
We define the graph-topology \emph{\iso}\footnote{The \iso of a graph (a.k.a. graph Cheeger constant) is the discrete analogue of the Cheeger isoperimetric constant. The analogy is described in detail in~\cite{buser1984bipartition}.}~\cite{mohar1989isoperimetric} as follows.
For a graph $G$ with set of nodes $V$ and non-empty subset $S\subseteq V$, let $\partial S$ denote the set of links of $G$ with one end in $S$ and the other end in $\overline{S}=V\setminus S$. Then, the \iso of graph $G$ is defined~as
\begin{align*}
i(G) &= \min_{\substack{S\subseteq V:\\|S|\leq |V|/2}} \frac{|\partial S|}{|S|}.
\end{align*}
}

\mm{
We adopt~\cite{GilbertRSpodc2018} the customary definition of \emph{graph-conductance} of a graph $G$ as the minimum over all cuts of the ratio between the number of edges across the cut and the minimum sum of the degrees of the nodes in each partition. Specifically, for any subset $S\subseteq V$, let $Vol(S)=\sum_{v\in S} d_v$, where $d_v$ is the number of neighbors of $v$. Then,
\begin{align*}
\Phi(G) &= \min_{S\subset V} \frac{|\partial S|}{\min\{Vol(S),Vol(\overline{S})\}}.
\end{align*}
}

\mm{
Also, we adopt~\cite{SJ:count} the closely-related \emph{conductance} of a Markov chain $\mathbf{X}$ on a state space $[n]$ with transition matrix $P$ as  
\begin{align*}
\phi(P) =  \min_{0< |S| < n} \max\left\{\frac{ \sum_{i\in S, j\in \overline{S}} p_{ij} \pi_i }{\sum_{i\in S} \pi_{i} }, \frac{ \sum_{i\in \overline{S}, j\in S} p_{ij} \pi_i }{\sum_{i\in \overline{S}} \pi_{i} }\right\}. 
\end{align*}
where $p_{ij}$ are the entries of $P$ and $\vec\pi$ is the stationary distribution of $\mathbf{X}$. 
}

\mm{
We will apply the latter definition to analyze a diffusion process embedded in one of our algorithms on a network graph $G=(V,E)$ with state space $V$ and transition matrix $P$. 
In our analysis, $\mathbf{X}$ is finite, irreducible, and ergodic, and $P$ is doubly-stochastic. Thus, the stationary distribution $\vec\pi$ is unique and uniform. That is, $\pi_i=1/n$ for all $i\in V$.
Also, our diffusion process is such that the transition matrix is symmetric. 
Then, we simplify the later definition of conductance as follows.
}

\mm{
The \emph{conductance} of a Markov chain $\mathbf{X}$ on a state space $[n]$ with symmetric transition matrix $P$ and uniform stationary distribution is defined as 
\begin{align*}
\phi(P) =  \min_{S\subset V} \frac{ \sum_{i\in S, j\in \overline{S}} p_{ij} }{\min\{|S|,|\overline{S}|\}}. 
\end{align*}
where $p_{ij}$ are the entries of $P$. 
}



\mm{
In all the above notation, we drop the specification of the graph or transition matrix when clear from context.
}


\section{Related Work}
\label{sec:relwork}


\LE is one the most fundamental and well-studied problems in distributed computing. The related literature includes a myriad of network topologies, communication models, and classes of algorithms, providing bounds in time and message complexity (e.g.~\cite{
LEoriginal,
awerbuch1987optimal,
peleg1990time,
KuttenJACM,
KuttenTCS,
frederickson1987electing,
GilbertRSpodc2018,
NO:unifleader,
jurdzinski:leader,
kowalski2013leader,
chlebus2012electing,
jurdzinski2012distributed
}).
However, even intuitive bounds for \LE randomized algorithms have not been proved until recently. In this section, we overview previous work on bounding time and message complexity in the \congest model, for the specific case of randomized algorithms that solve implicit \LE in synchronous networks of arbitrary topology, with $n$ un-labeled nodes, $m$ links, diameter $D$, conductance $\phi$, and random-walk mixing time $t_{mix}$. We summarize the most relevant bounds and the results in this work in Table~\ref{table:relwork}.

\mm{In a 2015 J. ACM paper,} Kutten et al.~\cite{KuttenJACM} showed the existence of graphs such that any universal\footnote{Algorithms that work for all graphs and all node ID assignments.} \LE algorithm that succeeds with at least some constant probability requires $\Omega(m)$ messages in expectation and $\Omega(D)$ time with constant probability, even if $n$, $m$, and $D$ are known. They also present $O(m)$ expected messages and $O(D)$ time algorithms for different assumptions of network-characteristics knowledge and success probabilities. (Refer to Table~\ref{table:relwork} for details.)
In separate work~\cite{KuttenTCS}, the same authors focused on complete and non-bipartite networks. 

\mm{In a PODC 2018 paper,} for known network size $n$, Gilbert et al.~\cite{GilbertRSpodc2018} showed
$O( t_{mix}\sqrt{n} \log^{7/2} n)$ messages are enough to solve implicit \LE with high probability in $O(t_{mix} n^2)$ time.
They also show that there exist graphs such that
$\Omega(\sqrt{n}/\phi^{3/4})$ expected messages are needed to solve the problem with probability at least $1-o(1)$, even if nodes have unique IDs and know the size of the network $n$. 
Knowledge of other network characteristics such as conductance, mixing time, or diameter is not needed.
Given the bounds on mixing time of $1/\phi \leq t_{mix} \leq 1/\phi^2$~\cite{montenegro2006mathematical},
aside from $polylog$ factors, these results leave a gap of at least $1/\phi^{1/4}$ and at most $1/\phi^{5/4}$. 
%
The algorithm presented breaks the $\Omega(m)$ lower bound of~\cite{KuttenJACM} for well-connected graphs. 
In the same work, for the case when the network size is unknown, it is proved that for any 2-connected graph there exist an ID assignment and port mapping such that $\Omega(m)$ expected messages are needed to solve implicit \LE with probability at least some constant.

Some of the results above are extended to other problems, such as Broadcast, tree construction and explicit \LE, once a leader has been elected.

\begin{table}[htpb]
\centering
\resizebox{\columnwidth}{!}{%
\begin{tabular}{c||c|c|c|c}
known & succes wp $1$ & succes whp & succes wp $1-o(1)$ & succes wp constant \\
\hline
\hline
\rule{0pt}{4ex}
$n,D$ & \begin{tabular}{c}$O(m)$ exp. msgs,\\ $O(D)$ exp time~\cite{KuttenJACM}\end{tabular} & & & \\[.1in]
\hline
\rule{0pt}{6ex}
\begin{tabular}{c}$n,\Phi,$\\$t_{mix}$\end{tabular} & & \begin{tabular}{c}{\bf $\widetilde{O}(\sqrt{nt_{mix}/\Phi})$ msgs,}\\ {\bf $O(t_{mix}\log^2 n)$ time}\\ {\bf {[}this work{]}} \end{tabular} &  & \\[.1in]
\hline
\rule{0pt}{4ex}
\multirow{2}{*}{$n$} & & \begin{tabular}{c} $O( t_{mix}\sqrt{n} \log^{7/2} n)$ msgs, \\$O(t_{mix} \log^2 n)$ time~\cite{GilbertRSpodc2018} \end{tabular} & \begin{tabular}{c}$\exists G$ : $\Omega(\sqrt{n}/\phi^{3/4})$\\ exp. msgs~\cite{GilbertRSpodc2018}\end{tabular} & \begin{tabular}{c}$\exists G$ : $\Omega(m)$\\ exp. msgs~\cite{KuttenJACM} \end{tabular} \\[.1in]
\cline{2-5}
\rule{0pt}{4ex}
& &  \begin{tabular}{c}$O(m \min(\log\log n, D))$\\ exp. msgs, $O(D)$ time~\cite{KuttenJACM}\end{tabular} &  & \begin{tabular}{c}$\exists G$ : $\Omega(D)$ time~\cite{KuttenJACM}\end{tabular} \\[.1in]
\cline{2-5}
\rule{0pt}{4ex}
& & \begin{tabular}{c}$O(m+n\log n)$ msgs,\\ $O(D\log n)$ time~\cite{KuttenJACM}\end{tabular} &  & \begin{tabular}{c}$O(m)$ exp. msgs,\\ $O(D)$ time~\cite{KuttenJACM}\end{tabular} \\[.1in]
\hline
\rule{0pt}{4ex}
- & & \begin{tabular}{c} {\bf $O(n^{4(2+\epsilon)} m \log^5 n)$ msgs,}\\ {\bf $O(n^{4(2+\epsilon)}\log^5 n)$ time}\\ {\bf {[}this work{]} (*)} \end{tabular} & & \begin{tabular}{c}$\forall$ 2-connected $G$ : $\exists$ labeling :\\ $\Omega(m)$ exp. msgs~\cite{GilbertRSpodc2018}\\$\forall T(n) : \nexists$ LE alg\\ in time $T(n)$  {\bf {[}this work{]}}\end{tabular} \\[.1in]
\hline
\rule{0pt}{4ex}
\mm{$i(G)$} & & \mm{\begin{tabular}{c} {\bf $O(\frac{n^{4(1+\epsilon)}}{i(G)^2} m \log^5 n)$ msgs,}\\ {\bf $O(\frac{n^{4(1+\epsilon)}}{i(G)^2}\log^5 n)$ time}\\ {\bf {[}this work{]} (*)} \end{tabular}} & & 
\end{tabular}%
}
\vspace{.1in}
\caption{Summary of relevant previous work and results in this work for randomized \sLE and (*) for randomized \eLE.}
\label{table:relwork}
\vspace{-.2in}
\end{table}


\section{Known Network Size}
\label{known-n}
In this section we present an \sLE algorithm for Anonymous Networks where 
the network size $n$ is known. 

Let $c>0$ be a sufficiently large constant (to satisfy all needs of the analysis).
We also assume that the parameter $x$ of the algorithm is at most some polynomial of $n$, and thus $\log x = O(\log n)$.

\mm{
A description of our algorithm for known 
number of nodes $n$, mixing time of the network $t_{mix}$, conductance $\Phi$, and a parameter $x$ \dk{and a constant $c>0$ that influence} correctness and communication complexity, follows. 
Pseudocode details can be found in Algorithm~\ref{alg:knownN}, 
\dk{and the auxiliary procedures in Algorithms~\ref{alg:knownN_broadcast},~\ref{alg:knownN_broadcast-1},~\ref{alg:knownN_broadcast-2}
and~\ref{alg:knownN_sub}}.
Given that we prove in our analysis asymptotic upper bounds, for the algorithm is enough to have linear upper bounds on $n$, $t_{mix}$, and $\Phi$. Nevertheless we use the exact values in our presentation for simplicity. 
} 

\begin{description}
\item[Selecting random IDs:]
Each node chooses its ID at random from the set $\{1,\ldots,n^4\}$.
\item[Selecting candidate nodes locally:]
Each node chooses to be a candidate independently with probability $(c \log n)/n$.
This means, and we will show this in the analysis, that there will be at most $4c\log n$ selected candidates whp.
\item[Candidate nodes span their territories by using Cautious broadcasts:]
Each candidate node initiates Cautious broadcast,
\dk{described later.}
\mm{Thus, each node is involved in multiple parallel executions of Cautious broadcast, call its number $y$.}
Time is partitioned into super-rounds, each containing $4c\log n$ subsequent rounds:
each node involved in a Cautious broadcast
for some ID dedicates at most one round in a super-round
to execute a corresponding step of \dk{the independent Cautious broadcast for this ID.} 
More precisely, if a node is involved in $y\le 4c\log n$ \dk{different parallel executions of Cautious broadcast methods for different IDs,}
i.e., initiated or received messages in those executions, it assigns its first $y$ rounds of the super-round to them 1-to-1 and runs
separately. In case $y>4c\log n$, \mm{the} nodes could assign arbitrary $4c\log n$ executions to available rounds --
we will show in the analysis that there are at most $4c\log n$ such \dk{parallel} Cautious broadcasts, whp,
therefore no node will ever need to assign arbitrarily \mm{a} \dk{subset of} executions to rounds, and thus each of these executions will be run correctly in the number of super-rounds upper bounded by \mm{time} complexity of \mm{the} Cautious broadcast \dk{method, whp.}
%
We call the nodes informed by a candidate (ID) during the execution of Cautious broadcast
its {\em broadcast territory}.
%
\mm{We omit these time-partitioning details in the pseudocode in Algorithm~\ref{alg:knownN} for clarity \dk{(they only cause $O(\log n)$ delay)}.}
\item[Candidate probe territories by random walks, larger wins:]
Each candidate node issues $x$ independent random walks of length $c\cdot t_{mix}\log n$, carrying its ID.
\dk{A random walk, used by our algorithm, propagates 
the source ID of the walk to a randomly (uniformly) selected neighbor of the current walk node,
while staying in the same node with probability $1/2$.}
If two or more walks meet at the same node, a decision about next destination is made independently for each walk.
In order to be implemented in the CONGEST model, note that there are only $O(\log n)$ different IDs of the walks, whp,
therefore by each link there is only sent an information about ID of the walk and the number of copies of the walk
that chose that link to propagate (which is encoded by $\log x = O(\log n)$ bits.
Further, 
once two different IDs meet, the smaller of the IDs is substituted by the larger one,
so that at most one ID is sent by a link each round. 
Each visited node stores the largest random walk ID ever seen.
\dk{A pseudo-code of the random-walk function is \dk{a part of} Algorithm~\ref{alg:knownN_sub}.}
\item[\dk{Convergecast of winning candidate IDs along the spanning tree of each broadcast territory:}]
\mm{
For each candidate, each node in its tree spanned in the execution of Cautious broadcast
in the beginning of the algorithm sends to its parent the largest walk ID seen.}
This is repeated $c\cdot t_{mix} \log n$ times, thus at the end the candidate \dk{(the source of that previously done cautious-broadcast)}
gets the largest walk ID that hit its \dk{broadcast} territory. 
\dk{Note that the partition of time into super-rounds, as was done to accommodate parallel executions of Cautious broadcast in the beginning of the algorithm, is not required, as during the convergecast a node passes only the largest walk ID ever seen.}
The candidate who did not hear a bigger candidate ID becomes a leader.
\dk{A pseudo-code of the convergecast function is \dk{a part of} Algorithm~\ref{alg:knownN_sub}.}
\end{description}

\noindent
{\bf Procedure Cautious broadcast}
is as follows.
The broadcast source performs a broadcast of its ID by spanning a tree, \dk{in a distributed way,} in 
\dk{$c\cdot t_{mix}\log n$ rounds,} 
using randomization for choosing new neighbors \dk{but only in sparse branches. More precisely,} 
in each round $t$ of the broadcast, each node $w$ who received the source ID maintains the following
knowledge and takes action accordingly:
\begin{itemize}
\item
its parent (originally it is the node who sent the source ID to $w$ as first, and in case of many such nodes -- one of them arbitrarily selected by $w$); 
the parent cannot be revoked until the end of the broadcast;
\item
its children -- each node who sets $w$ as its parent sends a message to $w$ 
confirming that it has chosen $w$ as its parent; then $w$ adds it to the set of its children; 
\item
the confirmed number of nodes in its subtree 
(i.e., nodes from whom the ``parent'' relation leads to $w$),
where confirmed means the sum of confirmed numbers obtained from all the received messages from its children \dk{plus one};
once its confirmed number exceeds a threshold $2^i$, for any $i\ge 0$, \mm{node $w$} sends this number to its parent and changes its searching status from active to passive, see \dk{the bullet} below;
\item
its searching status \dk{and, optionally, re-activation prompt: the status} could be active, which means that it can send a message to a randomly selected neighbor
(among ones to which it has not sent/received a message so far), or passive, in which case it does not send any message to a new
neighbor;
once the node receives a message from its parent to re-activate \dk{(so called re-activation prompt)}, it changes its status to active and sends re-activate message to its 
children;
a node also sends re-activate messages to its children from which it has received a new confirmed number,
if the new number did not raised its own confirmed number above the next threshold $2^i$ (this way the parent confirms
that the change on the number of nodes in specific subtree is legitimate, in the sense that it has not increased substantially the number
of nodes in the higher-level tree);
once a node reaches a threshold $x t_{mix}\Phi$ \mm{on the confirmed number of nodes in its subtree}, it sends messages to all its children (and the parent), and switches to passive mode
until the end of \dk{this execution of Cautious broadcast.}
\end{itemize}
\dk{A pseudo-code of Cautious broadcast could be find in Algorithms~\ref{alg:knownN_broadcast},~\ref{alg:knownN_broadcast-1} and~\ref{alg:knownN_broadcast-2}.}


\begin{algorithm*}[htbp]
\caption{\sLE algorithm for each node. All variables are global to all methods. $c>0$ is a constant.}
\label{alg:knownN}
\DontPrintSemicolon
\SetKwFor{round}{in one communication round:}{}{}
\SetKwFunction{le}{leader-election}
\SetKwFunction{cb}{cautious-broadcast}
\SetKwFunction{ccb}{candidate-cautious-broadcast}
\SetKwFunction{nccb}{non-candidate-cautious-broadcast}
\SetKwFunction{rw}{random-walk}
\SetKwFunction{cc}{convergecast}
\SetKwProg{myalg}{Procedure}{}{}
\SetKwProg{myfunc}{Function}{}{}

\myalg{\le{}}{
	$ID\gets i$ with probability $1/n^4$, for $i \in\{1,\dots,n^4\}$\;
	$
	candidate \gets
	\begin{cases}
		true, \textrm{ with probability $(c\log n)/n$,}\\ 
		false, \textrm{ with probability $1-(c\log n)/n$}.\\
	\end{cases}
	$\;
	\cb{}\;
	\rw{}\;
	\cc{}\;
	$leader \gets (ID=ID_{\max})$\;
}

\end{algorithm*}

\begin{algorithm*}[htbp]
\caption{Cautious-broadcast method.  Port numbers of this node: $1\dots N$. 
}
\label{alg:knownN_broadcast}
\DontPrintSemicolon
\SetKwFor{round}{in one communication round:}{}{}
\SetKwFunction{le}{leader-election}
\SetKwFunction{cb}{cautious-broadcast}
\SetKwFunction{rw}{random-walk}
\SetKwFunction{cc}{convergecast}
\SetKwFunction{pr}{process-receptions}
\SetKwFunction{pt}{prepare-transmissions}
\SetKwProg{myalg}{Procedure}{}{}
\SetKwProg{myfunc}{Function}{}{}

\setcounter{AlgoLine}{0}
\myfunc{\cb{}}{  
	\lIf{candidate}{$status\gets active$, $source \gets ID$}
	\lElse{$status\gets passive$, $source \gets null$}
	$parent\gets null$, $children\gets\emptyset$, 
	$threshold\gets 1$\; 
	$avail\gets \{1\dots N\}$, 
	$size[1\dots N]\gets 0$,
	$status[1\dots N]\gets passive$\;
	$\rec[1\dots N]\gets null$, 
	$\trans[1\dots N]\gets null$\;
	\For{$c~t_{mix} \log n$ synchronous communication rounds}{
		\If{$status\neq stop$}{
			\pr{}\;
			\pt{}\;
			\tcp{synchronous concurrent communication through all ports}
			\For{each $v\in [N]$}{
				\lIf{$\trans[v]\neq null$}{transmit $\trans[v]$ through port $v$}
			}
			\For{each $v\in [N]$}{
				$
				\rec[v] \gets
				\begin{cases}
					null, \textrm{ if no message received through port $v$,}\\ 
					\textrm{message received,  otherwise}.\\
				\end{cases}
				$\;
			}
		}
	}
}

\end{algorithm*}

\begin{algorithm*}[htbp]
\caption{Cautious-broadcast auxiliary method.  Port numbers of this node: $1\dots N$. 
}
\label{alg:knownN_broadcast-1}
\DontPrintSemicolon
\SetKwFor{round}{in one communication round:}{}{}
\SetKwFunction{le}{leader-election}
\SetKwFunction{cb}{cautious-broadcast}
\SetKwFunction{rw}{random-walk}
\SetKwFunction{cc}{convergecast}
\SetKwFunction{pr}{process-receptions}
\SetKwFunction{pt}{prepare-transmissions}
\SetKwProg{myalg}{Procedure}{}{}
\SetKwProg{myfunc}{Function}{}{}

\setcounter{AlgoLine}{0}
\myfunc{\pr{}}{  
	\eIf{candidate}{
			\For{each $v\in [N]$}{
					\lIf{$\rec[v] = \langle stop\rangle$}{
						$status\gets stop$
					}
					\ElseIf{$\rec[v]$ is some size}{
						$size[v]\gets \rec[v]$, 
						$status[v]\gets active$,
						$children\cup\{v\}$,
						$avail \gets avail\setminus\{v\}$\;
					}
			}
	}{
				\For{each $v\in [N]$}{
					\Switch{$\rec[v]$}{
						\lCase(\tcp*[f]{some tree node reached the threshold}){$\langle stop\rangle$}{
							$status\gets stop$
						}
						\lCase(\tcp*[f]{$v$ is the parent}){$\langle activate\rangle$}{
							$status \gets active$
						}
						\lCase(\tcp*[f]{$v$ is the parent}){$\langle deactivate\rangle$}{
							$status \gets passive$
						}
						\Case{some ID}{
							\If(\tcp*[f]{$v$ becomes the parent}){$source=null$}{
								$source \gets \rec[v]$, 
								$parent \gets v$, 
								$avail \gets avail\setminus\{v\}$,
								$status\gets active$,
								$status[v]\gets active$
							}
						}
						\Case(\tcp*[f]{$v$ is a child}){some size}{
							$size[v]\gets \rec[v]$, 
							$status[v]\gets active$,
							$children\cup\{v\}$,
							$avail \gets avail\setminus\{v\}$
						}
					}
				}
	}
}

\end{algorithm*}

\begin{algorithm*}[htbp]
\caption{Cautious-broadcast auxiliary method.  Port numbers of this node: $1\dots N$. 
}
\label{alg:knownN_broadcast-2}
\DontPrintSemicolon
\SetKwFor{round}{in one communication round:}{}{}
\SetKwFunction{le}{leader-election}
\SetKwFunction{cb}{cautious-broadcast}
\SetKwFunction{rw}{random-walk}
\SetKwFunction{cc}{convergecast}
\SetKwFunction{pr}{process-receptions}
\SetKwFunction{pt}{prepare-transmissions}
\SetKwProg{myalg}{Procedure}{}{}
\SetKwProg{myfunc}{Function}{}{}

\setcounter{AlgoLine}{0}
\myfunc{\pt{}}{  
			\lIf{$threshold\geq xt_{mix}\Phi$}{
				$status\gets stop$
			}
			\eIf{$status = stop$}{
				\For{each $v\in children$}{
					$\trans[v] = \langle stop\rangle$,
					$status[v]\gets stop$\;
				}
				\If{$\lnot candidate$}{
					$\trans[parent] = \langle stop\rangle$,
					$status[parent]\gets stop$\;
				}
			}{

				$\trans[1\dots N]\gets null$\;
				$subtree\gets 1$\;
				\lFor{each $v\in children$}{
					$subtree\gets subtree+size[v]$
				}
	\eIf{candidate}{
				\eIf{$subtree<threshold$}{
					\For{each $v\in children$ such that $status[v]\neq active$}{
						$\trans[v] = \langle activate\rangle$,
						$status[v]\gets active$\;
					}
				\If{$avail\neq\emptyset$}{
					choose a port $v\in avail$ uniformly at random\;
					$\trans[v] \gets \langle source\rangle$\;
				}
				}
				{
					$threshold\gets 2~threshold$\;
					\For{each $v\in children$ such that $status[v]\neq passive$}{
						$\trans[v] = \langle deactivate\rangle$,
						$status[v]\gets passive$\;
					}
				}
	}{
				$\trans[parent]\gets subtree$\;
				\eIf{$subtree<threshold$ {\bf and} $status=active$}{
					\For{each $v\in children$ such that $status[v]\neq active$}{
						$\trans[v] = \langle activate\rangle$,
						$status[v]\gets active$\;
					}
				\If{$avail\neq\emptyset$}{
					choose a port $v\in avail$ uniformly at random\;
					$\trans[v] \gets \langle source\rangle$\;
				}
				}
				{
					\lIf{$subtree\geq threshold$}{$threshold\gets 2~threshold$}
					\For{each $v\in children$ such that $status[v]\neq passive$}{
						$\trans[v] = \langle deactivate\rangle$,
						$status[v]\gets passive$\;
					}
				}
	}
			}
}

\end{algorithm*}

\begin{algorithm*}[htbp]
\caption{\sLE auxiliary methods.  Port numbers of this node: $1\dots N$. 
}
\label{alg:knownN_sub}
\DontPrintSemicolon
\SetKwFor{round}{in one communication round:}{}{}
\SetKwFunction{le}{leader-election}
\SetKwFunction{cb}{cautious-broadcast}
\SetKwFunction{rw}{random-walk}
\SetKwFunction{cc}{convergecast}
\SetKwProg{myalg}{Procedure}{}{}
\SetKwProg{myfunc}{Function}{}{}

\setcounter{AlgoLine}{0}

\setcounter{AlgoLine}{0}
\myfunc{\rw{}}{  
	$\#walks \gets$ 0,
	$ID_{\max}\gets ID$,
	$counter \gets 0$,
	$\trans[1\dots N]\gets\langle ID_{\max}, counter\rangle$\;
	\If{candidate}{
		\For{$x$ times}{
			choose some $v\in N$ uniformly at random\;
			$\trans[v].counter ++$\;
		}
	}
	\For{$c~t_{mix} \log n$ synchronous communication rounds}{
			\tcp{synchronous concurrent communication through all ports}
			\For{each $v\in [N]$}{
				\lIf{$\trans[v].counter \neq 0$}{transmit $\trans[v]$ through port $v$}
			}
			$\#walks\_rx \gets$ 0\;
			\For{each $v\in [N]$}{
				\If{some $\langle ID', counter'\rangle$ received through port $v$}{
					$\#walks\_rx \gets \#walks\_rx + counter'$\;
					\lIf{$ID'>ID_{\max}$}{$ID_{\max}=ID'$}
				}
			}
			\tcp{prepare transmissions}
			$\#walks \gets \#walks + \#walks\_rx$,
			$counter \gets 0$,
			$\trans[1\dots N]\gets\langle ID_{\max}, counter\rangle$\;
			\For{$i\gets\#walks$ {\bf down to} $0$}{
				$
				walk \gets
				\begin{cases}
					true, \textrm{ with probability $1/2$,}\\ 
					false, \textrm{ with probability $1/2$}.\\
				\end{cases}
				$\;
				\If{walk}{
					choose some $v\in N$ uniformly at random\;
					$\trans[v].counter++$\;
					$\#walks--$\;
				}
			}
	}
}

\setcounter{AlgoLine}{0}
\myfunc{\cc{}}{  
	\For{$c~t_{mix} \log n$ synchronous communication rounds}{
		\tcp{synchronous concurrent communication through all ports}
		\lIf{$\lnot candidate$}{transmit $ID_{\max}$ through port $parent$}
		\lIf{some $ID'>ID_{\max}$ received}{$ID_{\max}\gets ID'$}
	}
}

\end{algorithm*}

\subsection{Analysis}

Given that we 
consider
the CONGEST model, our bounds apply to messages as well as bits of communication modulo a logarithmic factor.

\begin{lemma}
\label{l:broadcast}
For a given parameter $x$, procedure Cautious broadcast
takes time $O(t_{mix}\log n)$;
it sends $\widetilde{O}(xt_{mix})$ messages/bits and informs $\widetilde{\Omega}(xt_{mix}\Phi)$ nodes~whp.
\end{lemma}

\begin{proof}
Time complexity and CONGEST model (i.e., $O(\log n)$ bits per point-to-point message)
follow directly from the description of the procedure.
The upper bound on the number of point-to-point messages sent follows from the following argumentation.
Due to control of doubling confirmed numbers
(of nodes in sub-trees), the total number of involved nodes never exceeds by factor $2$ the ultimate largest
threshold $xt_{mix}\Phi$. 
Therefore, the number of messages sent is bounded by this number 
\dk{multiplied by} 
the number
of messages in-between the informed nodes, which could be accounted as polylogarithmic per node
(indeed, a link is used a constant number of times per each change \dk{of the (exponential) thresholds of the confirmed numbers} at its end nodes).

It remains to prove that the procedure informs $\widetilde{\Omega}(xt_{mix}\Phi)$ nodes whp.
Observe that within the Cautious broadcast process we could distinguish a random walk process as follows:
it starts at the source node, and keeps going to a randomly selected node (if the selected node is already visited,
the walk goes there as well). One subtle issue is that the node in which the walk resides could be passive,
which means if the process selects randomly an unvisited neighbor, this is not what the broadcast does.
The broadcast has to wait, instead, but fortunately at most logarithmic number of rounds in a row.
Thus, in order to mimic $\Theta(t_{mix})$ steps of random walk, $\Theta(t_{mix}\log n)$ rounds
of Cautious broadcast protocols suffice. This means that a majority of nodes is within distance $\Theta(t_{mix}\log n)$
from the source. Using the analysis similar to the push process, c.f., \cite{kowalski2013estimating},
and \dk{applying the above argument} that until informing $x t_{mix}\Phi$ a node could be passive only by logarithmic number of rounds
while in the rest it performs a push-like action (with some coordination, but this even improves the speed of
source message propagation), we get that 
$\widetilde{\Omega}(xt_{mix}\Phi)$ nodes could be informed in $O(\Phi^{-1} \log n)\leq O(t_{mix}\log n)$ rounds;
here we used the fact that $\Phi^{-1} = O(t_{mix})$.
\end{proof}

\dk{In the remainder, fix 
$x=\widetilde{\Theta}\left(\sqrt{n\log n/(\Phi t_{mix})}\right)$.}

\begin{lemma}
\label{l:walks}
For $x=\widetilde{\Theta}\left(\sqrt{n\log n/(\Phi t_{mix})}\right)$, 
some walk with maximum ID visits some node in each candidate's broadcast territory whp.
\end{lemma}

\begin{proof}
It is enough to show for one territory (and the rest will follow by the union bound and whp).
Note that $x$ random walks initiated by a candidate will never change their ID.
Since they continue in $\Theta(t_{mix}\log n)$ rounds, the probability of not hitting the broadcast territory,
which contains $\widetilde{\Omega}(xt_{mix}\Phi)$ nodes whp \dk{(by Lemma~\ref{l:broadcast}),}
of a fixed candidate by any of these walks is at most 
$\left(1-\widetilde{\Omega}(xt_{mix}\Phi)/n\right)^x$, 
which is polynomially small 
for $x=\widetilde{\Theta}\left(\sqrt{n\log n/(\Phi t_{mix})}\right)$.
Thus, the hitting holds whp. 
\end{proof}

\begin{theorem}
\label{thm:knownn}
For $x=\widetilde{\Theta}\left(\sqrt{n\log n/(\Phi t_{mix})}\right)$, 
the leader election algorithm elects a unique leader and uses $\widetilde{O}(\sqrt{n t_{mix}/\Phi})$ point to point messages/bits of communication
in the CONGEST model with known (a linear upper bound on) $n$, whp.
It works in time $O(t_{mix}\log^2 n)$. 
\end{theorem}

\begin{proof}
The time complexity of $O(t_{mix}\log^2 n)$, as well as using $O(\log n)$ communication bits per point-to-point message whp,
follow directly from the description of the algorithm; the main contributor to time performance is Cautious broadcast
that has to be pipelined in logarithmic number of time-threads.

By Lemma~\ref{l:walks}, some walk with maximum ID visits some node in each candidate's broadcast territory whp.
This ID is propagated to each candidate via \dk{the convergecast along} its broadcast tree, which takes time and number of messages
not bigger than in the Cautious broadcast process, c.f., Lemma~\ref{l:broadcast}.
Recall that there are $\Theta(\log n)$ candidates, whp, thus the total communication is $\Theta(\log n)$ times
bigger than the performance for \dk{a single execution of Cautious} broadcast stated in Lemma~\ref{l:broadcast}.
Additionally, Lemma~\ref{l:walks} guarantees that exactly one candidate \dk{with biggest ID} is heard by all other candidates whp,
while the number of candidates is $\Theta(\log n)$ whp, therefore the leader election process is correct whp.
\end{proof}

Notice that the factor $\sqrt{n t_{mix}/\Phi}$ is asymptotically smaller than $\sqrt{n} t_{mix}$,
for $t_{mix} = \omega(1/\Phi)$, and never asymptotically larger, due to 
the known bounds $1/\Phi \leq t_{mix} \leq 1/\Phi^2$~\cite{montenegro2006mathematical}. Thus improving over the $O( t_{mix}\sqrt{n} \log^{7/2} n)$ messages in previous work~\cite{GilbertRSpodc2018}.



\section{Unknown Network Size}
\label{sec:unknownN}

\subsection{Impossibility of \sLE}
\label{sec:impossibility}

Recall that an algorithm solves the \sLE problem in time $T(n)$ with probability $p(n)$, where $n$ is an input size,
if for any integer $n>0$ and network $G$ of $n$ nodes, the probability that all nodes stop by time $T(n)$
with the same value and exactly one of them will have a flag raised with that value, is at least $p(n)$.
In this section, \mm{using a probabilistic pumping wheel technique} we prove the following result.

\begin{theorem}
\label{thm:imposs}
For any non-decreasing positive integer function $T(n)$ and any constant $0<c<1$,
there is no algorithm solving \sLE problem in time $T(n)$ with probability $c$, in the setting without known number of nodes $n$.
\end{theorem}

\begin{proof}

Suppose to the contrary that such an algorithm exists, call it $\cA$.
Consider an arbitrary positive integer $n$ and a cycle $C_n$ of $n$ nodes and $n$ edges.
Algorithm $\cA$ stops at all nodes of $C_n$ by time $T(n)$ with probability at least~$c$.
Without loss of generality assume that $\cA$ draws one random bit per round of communication. 
(If more random bits per round are used, the same argument can be extended to more outcomes.) 

Consider an execution of $\cA$ on $C_n$ in the first $T(n)$ rounds.
Starting from the initial state where nodes have no information (recall that nodes do not have labels and the network size is unknown), in each round a node makes decisions, based on the random bits drawn and the received states of its neighbors, to move to another state. With respect to the random bits drawn, by time $t$ the node may be in one of $2^t$ states. 
We call the states of all network nodes at a given time a \emph{configuration} of states. 
A configuration in $C_n$ where $\cA$ stops successfully at all nodes electing a leader is called a \emph{winning} configuration. 
By definition of $\cA$, the probability of ending at winning configuration is at least $c$, and that there are $2^{nT(n)}$ possible configurations. Hence, there must exist some winning configuration $\Gamma$ that occurs in $\cA$ with probability at least $c/2^{nT(n)}$.
Denote by $\Gamma_{|t}$ the part of configuration $\Gamma$ by round $t$.

Consider a cycle $C_N$, where $N$ will be defined later, where $\cA$ is executed. 
Let a path of length $2T(n)+2n$ in $C_N$ be called a \emph{witness}, 
the $2n$ nodes in the middle of a witness be called the \emph{core},
and each half of the core of size $n$ be called a \emph{segment}
(see Figure~\ref{fig:witness}).

\begin{figure}[htp]
\begin{center}
\psfrag{n}{$n$ nodes}
\psfrag{T(n)}{$T(n)$ nodes}
\psfrag{segment}{segment}
\psfrag{core}{core}
\psfrag{witness}{witness}
\includegraphics[width=0.5\textwidth]{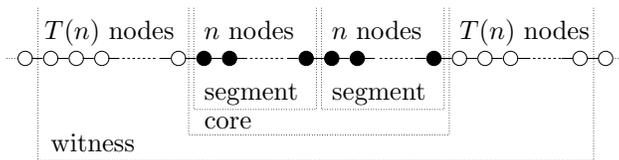}
\caption{Illustration of a witness.}
\label{fig:witness}
\end{center}
\end{figure}

We show now that, for $N$ large enough, after executing $\cA$ on $C_N$ for $T(n)$ rounds, the two segments in the core of some witness have configuration $\Gamma$, hence nodes in both segments stop each with an elected leader, with probability larger than $1-c$, which proves the theorem. 

Let $N$ be a multiple of $4T(n)+2n$ and consider $N/(4T(n)+2n)$ witnesses that are $2T(n)$-separated, that is, they are disjoint and between any pair
of consecutive witnesses there are at least $2T(n)$ nodes that do not belong to any considered witness.
Therefore, during the first $T(n)$ rounds of the execution of $\cA$ on $C_N$, the configurations on witnesses are 
independent.

Consider a single witness. Define a $t$-semi-core of the witness, for $0\leq t\leq T(n)$, as a set containing the core and all nodes of distance 
at most $T(n)-t$ from it; in particular, the $0$-semi-core is the witness itself, and the $T(n)$-semi-core is the core of the witness.
\mm{We prove the following invariant, illustrated in Figure~\ref{fig:invariant}.}
\begin{itemize}
\item[]
\emph{
For any $0\leq t\leq T(n)$, 
\mm{with probability at least $c/2^{nt}$, any node $v$ in the $t$-semi-core at distance $x\leq T(n)-t$ from the core has the same configuration in $\Gamma_{|t}$ as the node $v'$ at distance $x\bmod n$ from the center of the core.} 
} 

\end{itemize}
\begin{figure}[htp]
\begin{center}
\psfrag{b}{$t$-semi-core}
\psfrag{c}{core}
\psfrag{s}{segment}
\psfrag{wit}{witness}
\psfrag{v}{$v$}
\psfrag{w}{$v'$}
\psfrag{t}{$t$}
\psfrag{x}{$x$}
\psfrag{xmod}{$x \mod n$}
\includegraphics[width=\textwidth]{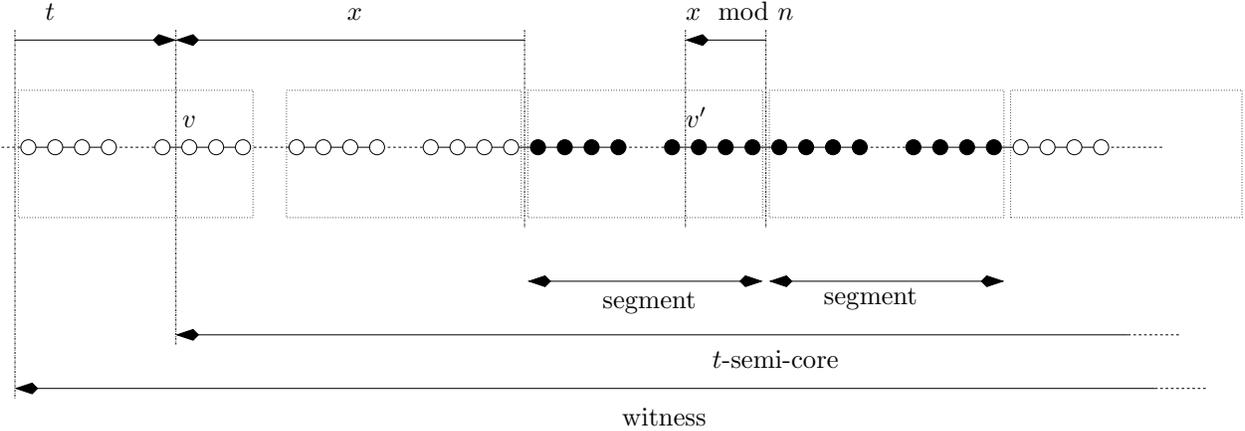}
\caption{Illustration of invariant.}
\label{fig:invariant}
\end{center}
\end{figure}

\mm{We prove the invariant by induction as follows.}
In the beginning all node configurations in the $0$-semi-core are the same. Then, if the invariant holds for some $0\leq t< T(n)$,
then we could extend the configuration of $t$-semi-core at all but the end nodes of this semi-core
to satisfy the invariant for $t+1$; this is because all these nodes receive messages from other nodes satisfying the invariant
for $t$, thus nodes mimicking (in the cyclic order modulo $n$, starting from the middle of the semi-core) the behavior of $\cA$ on $C_n$ at round $t+1$ of the protocol (leading to the partial configuration $\Gamma_{|t+1}$ on cycle $C_n$, and thus on corresponding
nodes of the $(t+1)$-semi-core, cyclically modulo $n$ starting from the middle of the semi-core).

It follows from the invariant for $T(n)$ that $n$ nodes in one segment of the core end up in configuration $\Gamma$,
so as the other $n$ nodes of the other segment of the core. Since $\Gamma$ is a winning configuration on $n$ nodes, each of these two segments stops electing one leader among themselves. This is an event violating the correctness during the considered
execution of $\cA$ on $C_N$ (for $T(n)$~rounds).

It remains to prove that the union of these events, over all the considered $N/(4T(n)+2n)$ witnesses,
holds with probability bigger than $1-c$. 
For $N = x (4T(n)+2n)$, 
we need to show that the complementary event holds with probability smaller than $c$, i.e.,
$\left(1-\left(\frac{c}{2^{nT(n)}}\right)^2\right)^x < c$. This is indeed implied by
$\exp\left(-\frac{c^2}{2^{2nT(n)}}x\right) < c$, which holds for
$x > \frac{\ln (1/c)}{c^2} 2^{2nT(n)}$.
%
%
That is, for $N=\left(1+\frac{\ln (1/c)}{c^2} 2^{2nT(n)}\right) (4T(n)+2n)$, with probability larger than $1-c$ there exists at least one witness in $C_N$ with two winning configurations $\Gamma$ after running $\cA$ for $T(n)$ steps, hence the core nodes stop with two leaders and the claim follows.
\end{proof}
\subsection{\eLEname}

In this section, we present a randomized algorithm that solves the explicit \eLE problem with high probability guarantees. 
\mm{Our algorithm does not use any network information. We show in our analysis its time and message complexity. We also show that if the isoperimetric number of the network graph is known tighter bounds are obtained.}

\mm{
A classic \LE technique is simple: the leader is the node that randomly chooses the smallest ID. The sample space must be some polynomial range on the number of nodes $n$. However, in a model with no network information, i.e. ``blindly'', it is not clear how to guarantee that only one node chooses the smallest ID. To overcome this challenge, in our algorithm nodes probe increasing estimates of network size, which eventually get close enough to $n$ to guarantee uniqueness. 
Each node chooses ID only once, but different nodes may do so for different estimates. 
Thus, to decide leadership each node compounds the chosen ID with the estimate used to choose it. 
The node with smallest ID, among those with largest estimate, is the leader. 
Our analysis shows that some node will not choose ID until the estimate is large enough.
That is, we use the estimate as a ``certificate'' of uniqueness. 
}

\mm{
Our technique to probe each estimate $k$ is based on averaging some node potential values, initially uneven.
That is, we use ``diffusion'' of potentials. 
The technique includes some ``thresholds'' and alarms to detect locally when $k$ is low, 
or otherwise choose ID in a range polynomial on $k$.
Thresholds, alarms, and the length of the diffusion process are functions of $k$, which may fail while $k$ is low.
Consequently, some nodes may choose an ID that may not be unique, and their decision is final.
Nevertheless, our analysis shows that, with large enough probability, until $k$ is close enough to $n$ to guarantee unique IDs, some nodes do not choose ID.
In other words, by the time all nodes have chosen ID, there is exactly one with smallest ID among those with largest estimate.
} 

\mm{In the following paragraphs we further detail \eLEname.} 
For clarity, we present the pseudocode in two parts as Algorithms~\ref{alg:unknownN} and~\ref{candAlg}. 
Thus, references to algorithm lines are given as $\langle algorithm\#\rangle.\langle line\#\rangle$.
\mm{The details of communication to fulfill the restrictions of the \congest model are left to the end for clarity.}

\mm{
The algorithm is structured in synchronous iterations over estimates of network size $k=2^i$ for $i=1,2,\dots$.
Each iteration corresponding to one estimate $k$, which we call iteration $k$, is conceptually divided in a \emph{\try phase} (Line~\ref{alg:unknownN}.\ref{itertry}) followed by a \emph{\dec phase} (Line~\ref{alg:unknownN}.\ref{idchoicedecision}).
}

\begin{description}
\item [\Try:]
\mm{
During the \try phase of an iteration $k$ nodes check whether $k$ is still too low with respect to $n$ to choose ID with various detection methods (described below). 
If a node $x$ fails to detect, 
either because $k$ is already large enough, 
or because $x$ does not detect $k$ as low and does not hear from other nodes detecting,
then $x$ chooses ID and stores it together with the value of $k$ used in this iteration (Line~\ref{alg:unknownN}.\ref{idchoice}).
}

\mm{
The \try phase is executed $f(k)$ times (for a function $f(\cdot)$ convenient for our analysis).
Using repeated executions is one of our detection methods, because at the beginning of each of these executions each node chooses to be a \emph{\cand} node with some probability $p(k)$ or \emph{\notcand} node with probability $1-p(k)$ (Line~\ref{alg:unknownN}.\ref{candchoice}), where $p(k)$ is chosen as convenient for our analysis.
Then, nodes inform their neighbors of their condition and, if in at least half of the executions a node $x$ hears from some \cand node, then $k$ is detected as low and $x$ does not choose ID for the current $k$ (Line~\ref{alg:unknownN}.\ref{idchoicedecision}).
} 

\mm{
After color choice, each execution of the \try phase is conceptually further divided in a  \emph{\dist phase} (Line~\ref{candAlg}.\ref{roundsleaderModif}) followed by a \emph{\diss phase} (Line~\ref{candAlg}.\ref{dissloop}).
}

\begin{description}
\item[\Dist:]
\mm{
During the \dist phase nodes share fractions of some potential values (Line~\ref{candAlg}.\ref{candPot}), which are initialized to $1$ for \notcand nodes and to $0$ otherwise (Line~\ref{candAlg}.\ref{initpot}). They also share their detection status, color status, etc. (Line~\ref{candAlg}.\ref{broadcast}.)
The number of rounds of \dist as well as the fractions are functions of $k$ as needed for the analysis.
Two other low-$k$ detection methods are used during \dist. 
Namely, having with too many neighbors with respect to $k$, or receiving from a neighbor that has detected low $k$  (Line~\ref{candAlg}.\ref{toomany}). 
Our last detection method is implemented at the end of the \dist phase: if the potential of a node $x$ is larger than some threshold $\tau(k)$, for a convenient function $\tau(\cdot)$, then $x$ detects $k$ as low  (Line~\ref{candAlg}.\ref{toomuch}). 
}

\item[\Diss:]
\mm{
Finally, in the \diss phase nodes broadcast their whole status. Namely, their color status and detection status, and the ID and estimate of the elected leader so far, if any (Line~\ref{candAlg}.\ref{statusbroad}). Local variables are updated as needed by the received status (Lines~\ref{candAlg}.\ref{updatebegin} to~\ref{candAlg}.\ref{updateend}). 
This broadcast runs for a number of rounds large enough to reach the whole network, should $k$ be close enough to $n$. 
Our analysis shows that at least one node chooses ID using a large enough $k$. Hence, all nodes in the network receive such information and know which is the leader. 
}
\end{description}

\item[\Dec:]
\mm{
In the \dec phase of an iteration $k$, each node $x$ that did not choose ID yet and did not detect $k$ as low in the \try phase chooses an ID independently and uniformly at random in a polynomial range on $k$ (Line~\ref{alg:unknownN}.\ref{idchoice}). The chosen ID is conceptually compounded with the value of $k$ used to choose it to be used as a certificate. Node $x$ stores the ID and certificate of the leader from $x$'s perspective. Initially its own, updating it as soon as $x$ receives a larger certificate or the same certificate with a smaller ID (Lines~\ref{candAlg}.\ref{updateminID} and~\ref{candAlg}.\ref{minID}). 
A boolean indicator of $x$ leadership is maintained accordingly (Line~\ref{alg:unknownN}.\ref{flag}). 
Node $x$ may have the leadership flag true for some time, but if $x$ certificate ($k$ used when ID was chosen) is low, at some point it will hear of a larger certificate from another node and will update the variables accordingly. 
}

\end{description}

\mm{
It remains to explain how the communication needed by the algorithm is implemented under the restrictions of the \congest model. 
Color status and detection status use $O(1)$ bits, whereas ID chosen and estimate to choose it need $O(\log n)$ bits. 
Potentials on the other hand need $\omega(\log n)$ bits in some rounds of communication. Hence, they are transmitted bit by bit. We leave out of the pseudocode this detail for clarity, but we take it into account in the analysis.  
}

\begin{algorithm*}[htbp]
\caption{\eLEname. Algorithm for each node. 
$f(k) = (4\sqrt{2} /(\sqrt{2}-1)^2) \ln (k^{1+\epsilon}/\xi)$,
$p(k) = \ln 2/k^{1+\epsilon}$,
\mm{and $0<\epsilon\leq1$.}
}
\label{alg:unknownN}
\DontPrintSemicolon
\SetKwRepeat{Do}{do}{while}
\SetKwBlock{Repeat}{repeat}{}
\tcp{global variables}
$k \gets 1$ \tcp*{running estimate of network size}
$leader \gets false$\tcp*{leader flag}
$id\gets nil$,
$K \gets nil$ \tcp*{id chosen and estimate when id chosen}
$id_{ldr}\gets nil$,
$K_{ldr} \gets nil$ \tcp*{id chosen and estimate when id chosen by leader}
	$status \gets$ new array of size $f(k)$ \tcp*{final status of each iteration} 
	$empty\gets$ new array of size $f(k)$ \tcp*{no \cand node detected flag of each iteration}
\Repeat{ \label{iter}
	$k \gets 2k$\;
	\tcp{\try phase}
	\For{$i = 1$ to $f(k)$}{ \label{itertry}
		$
		color \gets
		\begin{cases}
		\cand \textrm{ with probability $p(k)$,}\\ 
		\notcand \textrm{ with probability $1-p(k)$.}\\
		\end{cases}
		$\label{candchoice}\;
		$\langle q,c\rangle \gets$ \CMC$(color)$ \label{MMCcall}\; 
			$status[i]\gets q$\;
			$empty[i]\gets \lnot c$\;
	}
	\tcp{decision phase}
	\If{$id=nil$ {\bf and} \# of $true$ in $empty > f(k)/2$ {\bf and} \# of $probing$ in $status > 0$}{\label{idchoicedecision}
		$id         \gets$ choose an integer in 
		\mm{$[1,k^{4(1+\epsilon)}\log^4(4k)]$}
		uniformly at random, \label{idchoice}
		$K          \gets k$\;
		$id_{ldr} \gets id$,
		$K_{ldr}  \gets K$\label{update}\;
	}
	$leader \gets (K_{ldr} = K) \land (id_{ldr} = id)$\tcp*{update leader flag}\label{flag}
}
\end{algorithm*}


\begin{algorithm*}[htbp]
\caption{\CMC algorithm for each node. 
$N$ is the set of neighbors of this node.
$\vec\Phi,\vec q, \vec c, \vec{id_{ldr}}, \vec{K_{ldr}}$ are vectors such that $\Phi_i, q_i, c_i, id_{ldr_{i}}, K_{ldr_i}, \forall i\in N$.
\mm{$r(k)=(8k^{2(1+\epsilon)}/i(G)^2) \log (k^{2(1+\epsilon)}) + k^{1+\epsilon}\log (2k)$,}
$\tau(k) = 1-1/(k^{1+\epsilon}-1)$,
\mm{and $0<\epsilon\leq1$.}
For each link, transmissions of potentials are done one bit at a time in both directions. 
}
\label{candAlg}
\DontPrintSemicolon
\SetKwFunction{KwFn}{\CMC}
\SetKwProg{Fn}{Function}{}{end}
\Fn{\KwFn{$color$}}{
	$c \gets color=\cand$ \label{leaderinitp} \tcp*{existence of \cand nodes}
	$q\gets probing$ \tcp*{status: probing|low}
	\leIf(\tcp*[f]{potential}){\cand}{$\Phi\gets 0$}{$\Phi\gets 1$} \label{initpot}
	\tcp{\dist phase}
	\For{$round=1$ to $r(k)$}{  \label{roundsleaderModif}
		Broadcast $\langle\Phi,q,c,id_{ldr},K_{ldr} \rangle$ and Receive $\langle \vec\Phi,\vec q, \vec c, \vec{id_{ldr}}, \vec{K_{ldr}} \rangle$\;\label{broadcast}
		\If{$q=probing$ {\bf and} $|N|\leq k^{1+\epsilon}$ {\bf and} $\forall i\in N:q_i=probing$}{ \label{toomany}
			$\Phi\gets \Phi + \sum_{i\in N}\Phi_i/(2k^{1+\epsilon}) - |N|\Phi/(2k^{1+\epsilon})$ \label{candPot}	
		}
		\lElse{ \label{leadertoomany}
			$q\gets low$, \label{alarminsecondleader}
			$\Phi\gets 1$
		}
		\For{each $i\in N : id_{ldr_i}\neq nil$}{\label{updateminID}
			\lIf{$K_{ldr_i} > K_{ldr}$} {$id_{ldr} \gets id_{ldr_i}$, $K_{ldr} \gets K_{ldr_i}$}		
			\lElseIf{$K_{ldr_i} = K_{ldr}$ {\bf and} $id_{ldr_i} < id_{ldr}$}{$ id_{ldr} \gets id_{ldr_i}$}
		}
	} 
	\lIf{$\Phi> \tau(k)$}{ \label{toomuch}
		$q\gets low$,
		$\Phi\gets 1$ \label{candthreshold}
	} 
	\tcp{\diss phase}
	\For{$round=1$ to $k^{1+\epsilon}$}{  \label{dissloop}
		Broadcast $\langle q, c, id_{ldr}, K_{ldr} \rangle$ and Receive $\langle \vec q, \vec c, \vec{id_{ldr}}, \vec{K_{ldr}} \rangle$ \label{statusbroad}\;
		\For{each $i\in N$}{\label{updatebegin}
			\lIf{$q_i = low$}{$q\gets low$\label{lowdiss}}
			\lIf{$c_i = true$}{$c\gets true$}
		}
		\For{each $i\in N : id_{ldr_i}\neq nil$}{ \label{minID}
			\lIf{$K_{ldr_i} > K_{ldr}$} {$id_{ldr} \gets id_{ldr_i}$, $K_{ldr} \gets K_{ldr_i}$}		
			\lElseIf{$K_{ldr_i} = K_{ldr}$ {\bf and} $id_{ldr_i} < id_{ldr}$}{$ id_{ldr} \gets id_{ldr_i}$}
		}\label{updateend}
	} 
	\Return{$\langle q,c\rangle$}
}
\end{algorithm*}



\subsection{Analysis}

In this section, we analyze our \eLE algorithm for unknown $n$. 
Let $\ell$ be the number of \cand nodes after the random choices in Line~\ref{alg:unknownN}.\ref{candchoice}. That is, $0\leq \ell \leq n$.
Throughout the analysis, we use that for $0<x<1$ it is $\exp(-x/(1-x)) \leq 1-x \leq \exp(-x)$~\cite[\S 2.68]{book:mitrinovic}.
We denote as $\lceil\lceil x \rceil\rceil$ the smallest power of $2$ that is larger than $x$. 

%

The potential of nodes at the beginning of a round $r$ of the \dist phase of \CMC is denoted as a row vector $\vec{\Phi}_{r}$, 
where $\Phi_{r}[i]$ is the potential of node $i\in\{0,1,2,\dots,n-1\}$
(nodes are labeled only for the analysis),
and $||\vec\Phi_r|| = \sum_{i\in V} \Phi_{r}[i]$. 
The subindex $r$ is dropped when it is clear from context. 

\mm{For any given estimate $k$,} the fractions of potential shared are round independent (cf. Line~\ref{candAlg}.\ref{candPot}). 
Thus, the evolution of potentials during the \dist phase of \CMC can be characterized by a matrix 
$S=(s_{ij})_{i,j\in V}$, where
$s_{ij} = 1/(2k^{1+\epsilon})$ if $i\neq j$ and
$s_{ii} = 1-|N_i|/(2k^{1+\epsilon})$, where $N_i$ is the set of neighbors of node $i$, \mm{and $0<\epsilon\leq 1$.}
Consider the vector of potentials $\vec{\Phi}_{1}$ held by nodes at the beginning of a \dist phase where $k^{1+\epsilon}\geq |N_i|$ for all $i\in V$.
Then, for round $r>0$, $\vec{\Phi}_{r} = \vec{\Phi}_{1}S^{r}$ is the vector of potentials at the beginning of round $r$. 

Given that $S$ is stochastic, the characterization above can be seen as a Markov chain ${\bf X}$ where the state space is $V$ and the transition matrix is $S$. 
Thus, for a \dist phase where $k^{1+\epsilon}$ is larger than the number of neighboring nodes, we analyze the evolution of potentials leveraging previous work on convergence of Markov chains (e.g. for load balancing~\cite{RSW:loadBalancing, GM:loadBal}, gossip-based aggregates~\cite{BGPSgossip,KDGgossip}, and mass-distribution~\cite{FMT:aggJournal}). 

The main result of this section is the following. 
\begin{theorem}
\label{thm:unknownN}
\mm{For $0<\epsilon\leq 1$} and $0<\xi<1$, 
after running Algorithm~\ref{alg:unknownN} on a network with $n>1$ nodes
\begin{align*}
r(k) &= \mm{\frac{8k^{2(1+\epsilon)}}{i(G)^2} \log (k^{2(1+\epsilon)}) + k^{1+\epsilon}\log (2k),}
\ & \
p(k) &= \frac{\ln 2}{k^{1+\epsilon}}\ ,\\
\tau(k) &= 1-\frac{1}{k^{1+\epsilon}-1}\ , \ & \ 
f(k) &= \frac{4\sqrt{2}\ln (k^{1+\epsilon}/\xi) }{\left(\sqrt{2}-1\right)^2} \ ,
\end{align*}
the explicit \eLE problem is solved with probability at least $1-1/n^{\log(8/5)}-2\xi$, with 
\mm{$O(\frac{n^{4(1+\epsilon)}}{i(G)^2}\log^5 n)$ time and $O(\frac{n^{4(1+\epsilon)}}{i(G)^2} m \log^5 n)$ messages,}
where $m$ is the number of links.
\end{theorem}


\mm{The following corollary is a direct consequence of the theorem and the $i(G)\geq 2/n$ lower bound on the isoperimetric number.}

\mm{
\begin{corollary}
\label{cor:unknownN}
\mm{For $0<\epsilon\leq 1$} and $0<\xi<1$, 
after running Algorithm~\ref{alg:unknownN} on a network with $n>1$ nodes
\begin{align*}
r(k) &= 2k^{2(2+\epsilon)} \log (k^{2(1+\epsilon)}) + k^{1+\epsilon}\log (2k),
\ & \
p(k) &= \frac{\ln 2}{k^{1+\epsilon}}\ ,\\
\tau(k) &= 1-\frac{1}{k^{1+\epsilon}-1}\ , \ & \ 
f(k) &= \frac{4\sqrt{2}\ln (k^{1+\epsilon}/\xi) }{\left(\sqrt{2}-1\right)^2} \ ,
\end{align*}
the explicit \eLE problem is solved with probability at least $1-1/n^{\log(8/5)}-2\xi$, with 
$O(n^{4(2+\epsilon)}\log^5 n)$ time and $O(n^{4(2+\epsilon)} m \log^5 n)$ messages,
where $m$ is the number of links.
\end{corollary}
}

%
To prove Theorem~\ref{thm:unknownN} we first need to prove a series of intermediate results. 
\mm{
The general structure is the following.
Lemma~\ref{lemma:correct} shows that the \dist process converges to the average potential at each node. That is, for any arbitrarily small $\gamma>0$, after nodes carry on the diffusion for enough time, all nodes have a potential with relative error $\gamma$. 
Lemma~\ref{lemma:time} upper bounds the time for such convergence. The upper bound is a function of the conductance of the matrix underlying the diffusion process, the network size, and the relative error. 
Lemmas~\ref{lemma:nolowalarm}
and~\ref{lemma:allenoughempties}
show that two of our detection methods are correct. 
Namely, when the estimate is close to the network size and there are some \cand nodes, 
the potential of all nodes is below some threshold after \dist (Lemma~\ref{lemma:nolowalarm}),
and whp at least half of the iterations in the \try phase do not have \cand nodes (Lemma~\ref{lemma:allenoughempties}). 
Lemmas~\ref{lemma:somechooselater} and~\ref{lemma:somenonempty} show the correctness of our \try method.
That is, that until the estimate is close to the network size, there is always some node that does not choose ID
(Lemma~\ref{lemma:somechooselater}),
and that after the estimate is close to the network size, there is some iteration with some \cand node detected
(Lemma~\ref{lemma:somenonempty}). The latter lemma implies that potentials are not above threshold, which together with Lemma~\ref{lemma:somechooselater} implies that some node will choose ID within the appropriate range, as we argue in Theorem~\ref{thm:unknownN}.
Both lemmas relate the number of iterations of the \try phase with the probability of error.
}

First, we show that under the above conditions each node converges to the average over the whole network. 
\begin{lemma}
\label{lemma:correct}
Consider a \dist phase of \CMC where $k^{1+\epsilon}\geq |N_i|$ for all $i\in V$, \mm{and $0<\epsilon\leq1$}. 
For any $\gamma>0$ there exists $r(\gamma)\geq 0$ such that,
for any $r\geq r(\gamma)$ and all $i\in V$,~it~is
$$\frac{\left| {\Phi}_{r}[i] - ||\vec{\Phi}_{1}|| / n \right|}{||\vec{\Phi}_{1}|| / n} \leq \gamma \ .$$ 
\end{lemma}

\begin{proof}
If $k^{1+\epsilon}\geq |N_i|$ for all $i\in V$, the evolution of the \dist phase of \CMC can be characterized by a finite, irreducible Markov chain ${\bf X}$ with transition matrix $S$ as defined above. 
Given that the underlying graph has self-loops, ${\bf X}$ is aperiodic.
Then, by the fundamental theorem of Markov chains~\cite{book:motwani}, ${\bf X}$ is ergodic and it has a unique stationary distribution. 
Since $S$ is doubly stochastic, the system $\vec\pi=\vec\pi S$ admits the solution $\vec\pi=(1/n\dots 1/n)$.

Let $\vec\mu_r$ be the distribution at the beginning of round $r$. Given that the chain converges to the stationary distribution, we know that, for each $\gamma> 0$,  there is a $r(\gamma)\geq 0$ such that, for all $r\geq r(\gamma)$, $|\mu_r[i] - \pi[i]|/\pi[i] \leq \gamma$, for all $i\in V$.
Then, for any initial distribution $\vec{\mu}_1$ and for all $r\geq r(\gamma)$, the following holds.
\begin{align*}
\frac{1-\gamma}{n} &\leq \mu_r[i] \leq \frac{1+\gamma}{n} \textrm{, for all $i\in V$,}\\
\frac{1-\gamma}{n} &\leq \sum_{j\in V}\mu_1[j]S^r[j][i] \leq \frac{1+\gamma}{n} \textrm{, for all $i\in V$,}\\
\frac{1-\gamma}{n} &\leq S^r[j][i] \leq \frac{1+\gamma}{n} \textrm{, for all $j,i\in V$.}
\end{align*}

Where $S^r[j][i]$ is the value of row $j$ and column $i$ of $S^r$.
Using the latter bounds, 
for any node $i\in V$ we get the following 
bounds for $\Phi_{r}[i] = \sum_{j\in V} \Phi_{1}[j] S^r[j][i]$ :
\begin{align*}
\Phi_{r}[i] &
\leq \sum_{j\in V} \Phi_{1}[j]  \frac{1+\gamma}{n} = (1+\gamma)\frac{||\vec\Phi_{1}||}{n} \textrm{, and}\\
\Phi_{r}[i] &
\geq \sum_{j\in V} \Phi_{1}[j]  \frac{1-\gamma}{n} = (1-\gamma)\frac{||\vec\Phi_{1}||}{n} \ .
\end{align*}

Thus, 
the claim follows.
\end{proof}


Now we bound the time for convergence to the stationary distribution, i.e. the \emph{mixing time}.

\begin{lemma}
\label{lemma:time}
Consider a \dist phase of \CMC where $k^{1+\epsilon}\geq |N_i|$ for all $i\in V$, \mm{and $0<\epsilon\leq1$}. 
For any $\gamma>0$,
if the number of rounds is at least $r  \geq (2/\phi^2)\log (n/\gamma)$, 
where $\phi$ is the conductance of the graph underlying ${\bf X}$,
the distribution $\vec\mu_{r}$ is such that, for any $i\in V$, it is
$$\frac{|\mu_{r}[i]-1/n|}{1/n} \leq \gamma \ .$$ 
\end{lemma}

\begin{proof}
Customarily, we bound the mixing time as a function of the second eigenvalue of the transition matrix of the Markov chain ${\bf X}$: 
we have that (cf. Proposition 3.1 and 3.2 in~\cite{SJ:count})
$\max_{i,j \in V} \frac{|S^r[i][j]-\pi[j]|}{\pi[j]}$ is at most
\begin{align}
\frac{(\max\{|\lambda_i| : 1\leq i \leq N-1\})^r}{\min_{j\in V}\pi[j]} 
= \frac{\lambda_1^r}{\min_{j\in V}\pi[j]} 
\ . 
\label{eqeigen}
\end{align}
The latter is true because $s_{ii}\geq1/2$ for all $i\in V$, then ${\bf X}$ is aperiodic, and hence all eigenvalues are positive.

Using that $\pi[i]=1/n$ for all $i\in V$, we have that
$\max_{i,j \in V} n|S^r[i][j]-1/n| \leq n\lambda_1^r$. 
The inequality holds for any $i,j\in V$ and any initial distribution $\vec\mu_{1}$, thus we have that
\begin{align}
\forall i\in V : n|\mu_{r}[i]-1/n| \leq n\lambda_1^r. 
\label{abserr}
\end{align}

Bounding the right-hand side as $n\lambda_1^r \leq \gamma$ for any $\gamma>0$, we have that 
$r  \geq \log (n/\gamma)/\log (1/\lambda_1)$.
Using that $\log 1/\lambda_1 \geq 1-\lambda_1$ and that $1-\lambda_1 \geq \phi^2/2$ (cf. Lemma 3.3 in~\cite{SJ:count}), it is enough to have
$r  \geq (2/\phi^2) \log (n/\gamma)$.
Thus, the claim follows.
\end{proof}


We show now that for big enough $k^{1+\epsilon}$ no nodes have potential above some threshold.

\begin{lemma}
\label{lemma:nolowalarm}
If $k^{1+\epsilon} \geq 2n+1$, $r\geq (2/\phi^2) \log (k^{2(1+\epsilon)})$, $\ell\geq1$, \mm{and $0<\epsilon\leq1$}, at the end of the \dist phase no individual node has potential larger than $\tau=1-1/(k^{1+\epsilon}-1)$. 
\end{lemma}

\begin{proof}


Given that $k^{1+\epsilon}\geq 2n+1 > n \geq |N_i|$ for all $i\in V$,
fixing $\gamma=1/k^{1+\epsilon}$ in Lemma~\ref{lemma:time} we have that 
for $r  \geq (2/\phi^2)\log (k^{2(1+\epsilon)}) \geq (2/\phi^2)\log (n k^{1+\epsilon})$
it is $|\mu_{r}[i]-1/n| \leq 1/(nk^{1+\epsilon})$
for all $i\in V$.
Thus, given that $||\vec\Phi_1|| = n-\ell$, after $r\geq (2/\phi^2) \log (k^{2(1+\epsilon)})$ rounds no node has potential larger than $(n-\ell)(1/n+1/(nk^{1+\epsilon}))$. 
We complete the proof showing that the latter is not larger than $\tau$. That is,
$(n-\ell)\left(\frac{1}{n}+\frac{1}{nk^{1+\epsilon}}\right) 
\leq 1-\frac{1}{k^{1+\epsilon}-1}$.

Given that $\ell\geq 1$, it is enough to prove
\begin{align*}
\frac{n-1}{n}+\frac{n-1}{nk^{1+\epsilon}} &\leq 1-\frac{1}{k^{1+\epsilon}-1}\\
\frac{1}{k^{1+\epsilon}-1}+\frac{1}{k^{1+\epsilon}} &\leq \frac{1}{n}.
\end{align*}

Using that $k^{1+\epsilon} \geq 2n+1$, the claim follows.
\end{proof}


The following lemma shows that for large enough $k$ the majority of the iterations in the \try phase do not have \cand nodes.

\begin{lemma}
\label{lemma:allenoughempties}
If $k^{1+\epsilon}\geq 2n+1$, $p(k)=\ln 2/k^{1+\epsilon}$, and $f(k) \geq c_1 \ln k^{1+\epsilon} / \zeta$ where $\zeta=\frac{1}{2\sqrt{2}} \left(1-\frac{1}{\sqrt{2}}\right)^2$, \mm{and $0<\epsilon\leq1$}, the number of iterations of \try phase in Algorithm~\ref{alg:unknownN} without \cand nodes detected is at least $f(k)/2$ with probability at least $1-1/n^{c_1}$, for any $c_1>0$.
\end{lemma}

\begin{proof}
For any iteration of the \try phase, if there is a \cand node in the network, all the other nodes will detect its presence due to connectivity and the number of rounds of communication, which is more than $k^{1+\epsilon}>2n+1$ (cf. Line~\ref{candAlg}.\ref{dissloop}). 
Therefore, to prove the claim, it is enough to prove that the number of iterations without \cand nodes is at least $f(k)/2$ with the desired probability. 

The probability of not having \cand nodes in an iteration of the \try phase is 
$(1-p(k))^n$.
Let the random variable $X$ be the number of iterations without \cand nodes. 
The expected number of iterations without \cand nodes is 
$E(X) = f(k)(1-p(k))^n$.
The random choices among iterations are independent. Thus, by Chernoff bound~\cite{book:mitzenmacher} we have that 
$Pr(X\leq (1-\delta)E(X)) \leq \exp\left(- \delta^2 E(X) / 2 \right)$, for $0<\delta<1$.
Making $\delta=1-1/(2(1-p(k))^n)$, 
which is positive for $p(k)=\ln 2/k^{1+\epsilon}$ and $k^{1+\epsilon}>2n+1$, we have that  
\begin{align*}
Pr(X\leq f(k)/2) \leq \exp\left(- \frac{1}{2} \left(1-\frac{1}{2\left(1-p(k)\right)^n}\right)^2  f(k)\left(1-p(k)\right)^n \right).
\end{align*}

Replacing that $\left(1-p(k)\right)^n\geq 1/\sqrt{2}$ for $k^{1+\epsilon}>2n+1$ and $p(k)=\ln 2/k^{1+\epsilon}$,
\begin{align*}
Pr(X\leq f(k)/2) 
\leq \exp\left(- \frac{1}{2\sqrt{2}} \left(1-\frac{1}{\sqrt{2}}\right)^2 f(k) \right). 
\end{align*}

Then, to prove the claim it is enough to prove that 
\begin{align*}
e^{- \zeta f(k)}
\leq \frac{1}{n^{c_1}}, \text{ for } c_1>0 \text{ and } \zeta=\frac{1}{2\sqrt{2}} \left(1-\frac{1}{\sqrt{2}}\right)^2.
\end{align*}
Using that $k^{1+\epsilon}\geq 2n+1$, the latter is true for 
$f(k) \geq c_1 \ln k^{1+\epsilon}/\zeta$.
Hence, the claim follows.
\end{proof}


We now show that some nodes do not choose ID until $k$ is large enough.

\begin{lemma}
\label{lemma:somechooselater}
If $p(k)=\ln 2/k^{1+\epsilon}$, 
\mm{$r(k) \geq k^{1+\epsilon}\log (2k) $} 
and $f(k) = c_1 \ln (k^{1+\epsilon}/\xi) / \zeta$ where $\zeta=\frac{1}{2\sqrt{2}} \left(1-\frac{1}{\sqrt{2}}\right)^2$, 
\mm{$c_1>2\zeta/\ln(4/e)$,}
$0<\xi<1$, \mm{and $0<\epsilon\leq1$} 
there exist nodes that do not choose ID until 
\mm{$k^{1+\epsilon}\log(4k) \geq n$} 
with probability at least $1-\xi$.
\end{lemma}

\begin{proof}
For any node $v\in V$, let $N_v^{(h)}$ be the set of nodes within $h$ hops of $v$, including $v$, and let $n_v^{(h)}=|N_v^{(h)}|$.
Let $\cE_v$ be the event that node $v$ does not detect \cand nodes during an iteration of the \try phase.
We have that for each $k=2,4,\dots$ and each node $v\in V$ it is $Pr(\cE_v) = (1-p(k))^n + (1-p(k))^{n_v^{(h)}} (1-(1-p(k))^{n-n_v^{(h)}}) = (1-p(k))^{n_v^{(h)}}$, for $h = r(k)+k^{1+\epsilon}$. 

Let $X_v$ be the number of times when $\cE_v$ occurs over the $f(k)$ iterations of the \try phase. 
Then, 
$\mathbf{E}(X_v) = f(k) (1-p(k))^{n_v^{(h)}} $.
These iterations are independent. Hence, by Chernoff bound~\cite{book:mitzenmacher}, for $\delta >0$, we have that 
$Pr(X_v\geq (1+\delta) \mathbf{E}(X_v))$ is smaller than
\begin{align*}
\left(\frac{1}{e}\right)^{\mathbf{E}(X_v)}\left( \frac{e}{1+\delta}\right)^{(1+\delta)\mathbf{E}(X_v)} \ .
\end{align*}
For $1+\delta = 1/(2(1-p(k))^{n_v^{(h)}})$, which implies $\delta>0$ because $(1-p(k))^{n_v^{(h)}} < (1-p(k))^{k^{1+\epsilon}} \leq 1/2$ for $p(k)=\ln 2/k^{1+\epsilon}$, it is
\begin{align*}
Pr(X_v\geq f(k)/2) 
< \left( 2e(1-p(k))^{n_v^{(h)}}\right)^{f(k)/2}.
\end{align*}
Given that 
\mm{$n_v^{(h)}\geq h=r(k)+ k^{1+\epsilon} \geq k^{1+\epsilon}\log(2k) + k^{1+\epsilon} = k^{1+\epsilon}\log(4k)$,} 
we have that 
\mm{
\begin{align*}
Pr(X_v\geq f(k)/2) 
< \left( 2e(1-p(k))^{\min\{n,k^{1+\epsilon}\log(4k)\}}\right)^{f(k)/2}.\label{eq:pr}
\end{align*}
}
For 
\mm{$k^{1+\epsilon}\log(4k) < n$} 
\mm{and $\epsilon\leq 1$}, 
replacing $p(k)=\ln 2/k^{1+\epsilon}$ we have 
\mm{
\begin{align*}
Pr(X_v\geq f(k)/2) 
&< 
\left( \frac{e}{4k}\right)^{f(k)/2}.
\end{align*}
}
Let $Y_v$ be a random variable indicating that $X_v\geq f(k)/2$ for some \try phase with estimate $k$ such that 
\mm{$k^{1+\epsilon}\log(4k) < n$.} 
Then, it is 
\mm{
\begin{align*}
Pr(Y_v=1) 
&< 
\sum_{i=1}^\alpha \left( \frac{e}{2^{i+2}}\right)^{f(2^i)/2}.
\end{align*}
}
where $\alpha$ is such that 
\mm{$2^{\alpha(1+\epsilon)}\log(2^{\alpha+2}) < n$} 
and 
\mm{$2^{(\alpha+1)(1+\epsilon)}\log(2^{\alpha+3}) \geq n$.}
%
\mm{We also have that
\begin{align*}
\sum_{i=1}^\alpha \left( \frac{e}{2^{i+2}}\right)^{f(2^i)/2} 
< \sum_{i=1}^\alpha \left( \frac{e}{2^{i+2}}\right)^{f(2)/2}
=  \left(\frac{e}{4}\right)^{f(2)/2}  \sum_{i=1}^\alpha \left(\frac{1}{2^{i}}\right)^{f(2)/2}
&<  \left(\frac{e}{4}\right)^{f(2)/2}. 
\end{align*}
}
Replacing $f(k) = c_1 \ln (k^{1+\epsilon}/\xi) / \zeta$ we get
\mm{
\begin{align*}
Pr(Y_v=1) 
&< \left( \frac{e}{4}\right)^{c_1 \ln (2/\xi) / (2\zeta)}
= \left( \frac{\xi}{2}\right)^{c_1 \ln (4/e) / (2\zeta)}.
\end{align*}
}
Let $Y=\sum_{v\in V} Y_v$. Then, the expected number of nodes that choose ID while 
\mm{$k^{1+\epsilon}\log(4k) < n$} 
is $\mathbf{E}(Y) = n Pr(Y_v=1)$.
By Markov's inequality~\cite{book:mitzenmacher}, we have that 
\mm{
$Pr(Y\geq n-1) < \mathbf{E}(Y)/(n-1)< 2 \left( \xi/2\right)^{c_1 \ln (4/e) / (2\zeta)}<\xi$, for $c_1>2\zeta/\ln(4/e)$.
}
Thus, the claim follows.
\end{proof}


Finally, we show that when $k$ is large enough so that the majority of the iterations in the \try phase do not have \cand nodes, there is still some iteration with \cand nodes. 

\begin{lemma}
\label{lemma:somenonempty}
If $2n+1\leq k^{1+\epsilon} \leq 4n$, 
\mm{$r(k) \geq k^{1+\epsilon}$,} 
$p(k)=\ln 2/k^{1+\epsilon}$, and $f(k) = c_1 \ln (k^{1+\epsilon}/\xi) / \zeta$ where $\zeta=\frac{1}{2\sqrt{2}} \left(1-\frac{1}{\sqrt{2}}\right)^2$, $c_1>0$, $0<\xi<1$, \mm{and $0<\epsilon\leq1$,} there exist some iteration of the \try phase in Algorithm~\ref{alg:unknownN} with \cand nodes detected with probability at least $1-\xi$, for $0<\xi<1$.
\end{lemma}
\begin{proof}
\mm{Given that $r(k)\geq k^{1+\epsilon} \geq 2n+1$,} 
all nodes are within reach of each other
\mm{in every iteration of the \try phase.} 
Then, the probability that no \cand nodes are detected in \emph{all} iterations of the \try phase is at most $(1-p(k))^{nf(k)}\leq \exp(-p(k)nf(k))$. Replacing $p(k)$, $f(k)$, and $n$ the probability is at most 
\begin{align*}
\exp\left(- \frac{\ln 2}{k^{1+\epsilon}} \cdot \frac{k^{1+\epsilon}}{4} \cdot \frac{c_1}{\zeta} \ln \frac{k^{1+\epsilon}}{\xi} \right) 
\leq
\left( \frac{\xi}{k^{1+\epsilon}} \right)^{\frac{c_1\ln 2}{2\zeta}} < \xi. 
\end{align*}
\end{proof}


We can now prove 
the main theorem of this section as follows.

\paragraph{Proof of Theorem~\ref{thm:unknownN}}

\begin{proof}
We verify first that the functions $r(k)$, $p(k)$, $\tau(k)$, and $f(k)$ defined meet the conditions of the lemmas used. 
Notice that $\tau(k)$ and $p(k)$ are as defined in Lemmas~\ref{lemma:nolowalarm},~\ref{lemma:allenoughempties}, and~\ref{lemma:somechooselater}. 
\mm{
Using that the fraction of potential shared in the diffusion process in Algorithm~\ref{candAlg} is $1/(2k^{1+\epsilon})$,
the isoperimetric number of the network graph $G$ 
is
$i(G) = \phi 2k^{1+\epsilon}$ (cf. definitions of conductance and isoperimetric number in Section~\ref{sec:model}).
Thus,
 $r(k)=(8k^{2(1+\epsilon)}/i(G)^2) \log (k^{2(1+\epsilon)}) + k^{1+\epsilon}\log (2k)$ fulfills the condition $r \geq (2/\phi^2) \log (k^{2(1+\epsilon)})$ of Lemma~\ref{lemma:nolowalarm},
as well as 
$r(k) \geq k^{1+\epsilon}\log (2k)$ 
in Lemma~\ref{lemma:somechooselater}
and 
$r(k) \geq k^{1+\epsilon}$ 
in Lemma~\ref{lemma:somenonempty}.}
Finally, 
for $\zeta=\frac{1}{2\sqrt{2}} \left(1-\frac{1}{\sqrt{2}}\right)^2$,
it is $f(k) = \ln (k^{1+\epsilon}/\xi) / \zeta
\geq c_1 \ln k^{1+\epsilon} /\zeta$
for $c_1=1>0$ and $\xi<1$, as required by Lemmas~\ref{lemma:allenoughempties},~\ref{lemma:somechooselater}, and~\ref{lemma:somenonempty}.

To prove correctness, we observe first that, by Lemma~\ref{lemma:somechooselater}, with probability at least $1-\xi$ there exist some nodes that choose ID after \mm{$k^{1+\epsilon}\log(4k)\geq n$.} Hence, those nodes choose ID uniformly at random in the range $1$ to at least \mm{$k^{4(1+\epsilon)}\log^4(4k) \geq n^4$} (cf. Line~\ref{alg:unknownN}.\ref{idchoice}), which guarantees that their ID's are unique (for the same $k$) with probability at least $1-1/(2n^2)$~\cite{luby:mis}.
Also, when $k^{1+\epsilon}\geq 2n+1$, by Lemma~\ref{lemma:allenoughempties} we know that all nodes have not detected \cand nodes in at least $f(k)/2$ iterations of the \try phase with probability at least $1-1/n$, and by Lemma~\ref{lemma:somenonempty} we know that when $2n+1\leq k^{1+\epsilon} \leq 4n$ there are some iterations with \cand nodes with probability at least $1-\xi$. Hence, by Lemma~\ref{lemma:nolowalarm} we know that no node has potential larger $\tau(k)$. Combining these facts, we know that not all nodes choose ID before \mm{$k^{1+\epsilon}\log(4k)\geq n$} and all nodes have chosen ID after $k^{1+\epsilon} > 4n$, and the variables $id_{ldr}$ and $K_{ldr}$ will be updated in Line~\ref{alg:unknownN}.\ref{update} with the same values for all nodes within the next iteration (Line~\ref{alg:unknownN}.\ref{iter}). Given that the minimum ID of the maximum $K$ is taken (Line~\ref{alg:unknownN}.\ref{idchoicedecision}), the elected node is unique and the \eLE problem is solved with probability at least $1-1/(2n^2)-1/n-2\xi\geq 1-1/n^{\log(8/5)}-2\xi$ for $n>1$.  

We compute the asymptotic running time by inspection of Algorithms~\ref{alg:unknownN} and~\ref{candAlg}. For each of the values of $k=2,4,\dots, 2\lceil\lceil (4n)^{1/(1+\epsilon)}\rceil\rceil$, and for each of the $f(k)$ iterations of the \try phase in Algorithm~\ref{alg:unknownN}, Algorithm~\ref{candAlg} is executed. The \dist phase in Algorithm~\ref{candAlg} takes $r(k)$ iterations. Each iteration $i=1,\dots,r(k)$ takes $i\log (2k^{1+\epsilon})$ rounds of communication because potentials are divided by $2k^{1+\epsilon}$ in each iteration and potentials are transmitted bit by bit to fulfill the restrictions of the \congest model. (The other values transmitted all have $O(\log n)$ bits.) The \dist phase is followed by the \diss phase that takes $k^{1+\epsilon}$ rounds of communication. Thus, overall, the \eLE problem is solved within time
\begin{align*}
\sum_{i=1}^{a} f(2^i) \left(
\log(2^{i(1+\epsilon)+1})\sum_{j=1}^{r(2^i)} j 
+2^{i(1+\epsilon)}\right)
&= \sum_{i=1}^{a} f(2^i) \left(
(i(1+\epsilon)+1) \frac{r(2^i)(r(2^i)+1)}{2} 
+2^{i(1+\epsilon)}\right),
\end{align*}
where 
$a=1+\left\lceil\frac{2+\log n}{1+\epsilon}\right\rceil$,
\mm{$r(2^i) = 2^{2i(1+\epsilon)} \frac{16i(1+\epsilon)}{i(G)^2} + (i + 1)2^{i(1+\epsilon)}$,}
and
$f(2^i) = \frac{4\sqrt{2}}{\left(\sqrt{2}-1\right)^2} \ln (2^{i(1+\epsilon)}/\xi) = \frac{4\sqrt{2}(1+\epsilon)i}{\xi\log e\left(\sqrt{2}-1\right)^2}$.

Replacing, the claimed time complexity follows. Taking into account that every node communicates with all neighbors in each round, the number of messages per round is in $O(m)$. Hence, the claimed message complexity follows. 
\end{proof}


\section{Conclusions and Open Problems}

Motivated by recent developments for massive Ad-hoc Networks embedded in the Internet of Things, 
we have studied randomized \LE in \ANs in two scenarios: with and without knowledge of network size.  
We presented and analyzed randomized protocols for \sLE and \eLE.
To the best of our knowledge ours is the first study of \eLE protocols, whereas our \sLE protocol improves over previous work in message complexity. 
Improving the polynomial complexities of the solution for \eLE problem, and filling in the missing fields 
in Table~\ref{table:relwork}, are interesting and challenging future directions.

\bibliographystyle{plainurl}
\bibliography{Comprehensive_2010}

\end{document}